\documentclass[a4paper,11pt]{article}
\usepackage{amssymb,amsmath,amsthm}
\usepackage[UKenglish]{babel}
\usepackage{microtype}
\usepackage{a4wide}
\usepackage{thm-restate}
\usepackage{complexity}
\usepackage{comment}
\usepackage{adjustbox}
\usepackage{hyperref} 
\hypersetup{colorlinks=true,citecolor=blue,linkcolor=blue,urlcolor=blue}
\usepackage{tikz}
\usepackage{tikz-cd}
\usetikzlibrary{calc}
\tikzset{curve/.style={settings={#1},to path={(\tikztostart)
    .. controls ($(\tikztostart)!\pv{pos}!(\tikztotarget)!\pv{height}!270:(\tikztotarget)$)
    and ($(\tikztostart)!1-\pv{pos}!(\tikztotarget)!\pv{height}!270:(\tikztotarget)$)
    .. (\tikztotarget)\tikztonodes}},
    settings/.code={\tikzset{quiver/.cd,#1}
        \def\pv##1{\pgfkeysvalueof{/tikz/quiver/##1}}},
    quiver/.cd,pos/.initial=0.35,height/.initial=0}

\definecolor{goldenrod}{RGB}{218, 165, 32}

\DeclareMathOperator{\ar}{ar}

\DeclareMathOperator{\PCSP}{PCSP}
\DeclareMathOperator{\Pol}{Pol}

\renewcommand{\D}{\ensuremath{\mathbf{D}}}

\newcommand{\nat}{\ensuremath{\mathbb{N}}}
\newcommand{\F}{\ensuremath{\mathbf{F}}}
\newcommand{\NAE}{\ensuremath{\mathbf{NAE}}}
\renewcommand{\M}{\ensuremath{\mathcal{M}}}

\newcommand{\N}{\ensuremath{\mathcal{N}}}
\renewcommand{\A}{\ensuremath{\mathbf{A}}}
\newcommand{\B}{\ensuremath{\mathbf{B}}}
\newcommand{\I}{\ensuremath{\mathbf{I}}}
\newcommand{\LO}{\ensuremath{\mathbf{LO}}}

\newcommand\Crestrict[2]{{
  \left.\kern-\nulldelimiterspace 
  #1 
  \right|_{#2} 
 }}

\newcommand\pmc{\ensuremath{\textsc{PMC}}}

\theoremstyle{plain}
\newtheorem{theorem}{Theorem}
\newtheorem{lemma}[theorem]{Lemma}
\newtheorem*{lemma*}{Lemma}
\newtheorem{proposition}[theorem]{Proposition}
\newtheorem*{proposition*}{Proposition}
\newtheorem{corollary}[theorem]{Corollary}
\newtheorem*{corollary*}{Corollary}

\theoremstyle{definition}
\newtheorem{definition}[theorem]{Definition}
\newtheorem{remark}[theorem]{Remark}
\newtheorem{example}[theorem]{Example}

\usepackage{marginnote}

\begin{document}

\author{Tamio-Vesa Nakajima\\
University of Oxford\\
\texttt{tamio-vesa.nakajima@cs.ox.ac.uk}
\and
Stanislav \v{Z}ivn\'y\\
University of Oxford\\
\texttt{standa.zivny@cs.ox.ac.uk}
}

\title{Linearly ordered colourings of hypergraphs\thanks{An extended abstract of
part of this work (with weaker both tractability and intractability results) appeared in the Proceedings of ICALP 2022~\cite{NZ22:icalp}. This project has received funding from the European Research Council (ERC) under the European Union's Horizon 2020 research and innovation programme (grant agreement No 714532). The paper reflects only the authors' views and not the views of the ERC or the European Commission. The European Union is not liable for any use that may be made of the information contained therein. This work was also supported by UKRI EP/X024431/1 and a Clarendon Fund Scholarship. For the purpose of Open Access, the authors have applied a CC BY public copyright licence to any Author Accepted Manuscript version arising from this submission. All data is provided in full in the results section of this paper.}}

\date{}
\maketitle

\begin{abstract}
A linearly ordered (LO) $k$-colouring of an $r$-uniform hypergraph 
assigns an integer from $\{1, \ldots, k \}$ to every vertex so that, in every edge, the
(multi)set of colours has a unique maximum. Equivalently, for $r=3$, if two vertices in
an edge are assigned the same colour, then the third vertex is assigned a
larger colour (as opposed to a different colour, as in classic
non-monochromatic colouring).
Barto, Battistelli, and Berg~[STACS'21] studied LO
colourings on $3$-uniform hypergraphs in the context of promise constraint
satisfaction problems (PCSPs).
We show two results. 

First, given a 3-uniform hypergraph that admits an LO $2$-colouring, one can find in polynomial time an LO $k$-colouring with 
$k=O(\sqrt[3]{n \log \log n / \log n})$.

Second, given an $r$-uniform hypergraph that admits an LO $2$-colouring, we establish \NP-hardness of finding an LO $k$-colouring for every constant uniformity $r\geq k+2$.
In fact, we determine  
relationships between polymorphism minions for all uniformities $r\geq 3$, which reveals a key difference between $r<k+2$ and $r\geq k+2$ and which may be of independent interest. 
Using the algebraic approach to PCSPs, we actually show a more general result establishing \NP-hardness of finding an LO $k$-colouring for LO $\ell$-colourable $r$-uniform hypergraphs for $2 \leq \ell \leq k$ and $r \geq k - \ell + 4$.
\end{abstract}

\section{Introduction}
\label{sec:intro}

The computational complexity of the \emph{approximate graph colouring}
problem~\cite{GJ76} is an outstanding open problem in theoretical computer
science. Given a $3$-colourable graph $G$ on $n$ vertices, is it possible to
find a $k$-colouring of $G$? On the tractability side, the current best is a
polynomial-time algorithm of Kawarabayashi and Thorup \cite{KT17} that finds a $k$-colouring
with $k=k(n)=n^{0.199}$ colours. On the intractability side, the
state-of-the-art for constant $k$ has only recently been improved from
$k=4$, due to Khanna, Linial, and Safra~\cite{KLS00} and Guruswami and
Khanna~\cite{GK04} to $k=5$, due to Barto, Bul\'{i}n, Krokhin, and Opr\v{s}al~\cite{BBKO21}. The authors of~\cite{BBKO21} introduced a general algebraic
methodology for studying the computational complexity of so-called promise
constraint satisfaction problems (PCSPs). Going beyond the work in~\cite{BBKO21}, for
graphs with a promised higher chromatic number than three, the current best
intractability results for constantly many extra colours is due to
Wrochna and \v{Z}ivn\'{y}~\cite{WZ20}, building on the work of
Huang~\cite{Hua13}.
Hardness of approximate graph colouring has been established under stronger assumptions. In particular, 
Dinur, Mossel, and Regev~\cite{Dinur09:sicomp} showed NP-hardness under a non-standard variant of the UGC Conjecture of Khot~\cite{Khot02stoc} and Guruswami and Sandeep~\cite{GS20:icalp} (building on~\cite{WZ20}) showed NP-hardness under the $d$-to-$1$ conjecture~\cite{Khot02stoc} for any fixed $d$.

The situation is much better understood for the approximate \emph{hypergraph} colouring
problem with the classic notion of a colouring leaving no edge monochromatic.
A celebrated result of Dinur, Regev, and Smyth established
that finding a $k$-colouring of a $3$-uniform hypergraph that is $\ell$-colourable
is \NP-hard for every constant $2\leq \ell\leq k$~\cite{DRS05} (and this also
implies the same result on $r$-uniform hypergraphs for every constant uniformity $r\geq 3$). This was also recently proved, with different techniques, by Wrochna~\cite{Wrochna22}.

Different variants of approximate hypergraph colourings, such as rainbow
colourings, were studied, e.g.
in~\cite{ABP20,BG16-graph,GL18,GS20:rainbow,BG21:talg}, but most complexity
classifications related to these problems are open.
Some intractability results are also known for colourings with a super-constant
number of colours.  For graphs, conditional hardness was established by
Dinur and Shinkar~\cite{Dinur10:approx}. For hypergraphs, intractability results were
obtained by Bhangale~\cite{Bhangale18:icalp} and by Austrin, Bhanghale, and Potukuchi~\cite{ABP19}.

Barto, Battistelli, and Berg have recently studied systematically a certain type of PCSPs on non-Boolean domains and
identified a very natural 
variant of $k$-colourings of $3$-uniform hypergraphs, called \emph{linearly ordered} (LO)
$k$-colourings~\cite{Barto21:stacs}. A $k$-colouring of a $3$-uniform hypergraph
with colours $[k] = \{1,\ldots,k\}$ is an LO colouring if, for every edge, it holds
that, if two vertices are coloured with the same colour, then the third vertex is
coloured with a larger colour. (In the classic non-monochromatic colouring, the
requirement is that the third vertex should be coloured with a \emph{different}
colour, but not necessarily a larger one.) An LO $2$-colouring is thus 
 a ``$1$-in-$3$'' colouring. Barto et al.~asked whether finding an LO
 $k$-colouring of a $3$-uniform hypergraph is \NP-hard for a fixed $k\geq 3$ if the input hypergraph is
 promised to admit an LO $2$-colouring.

\paragraph{Contributions}
While we do not resolve the question raised in~\cite{Barto21:stacs}, we obtain non-trivial results, both positive (algorithmic) and negative (hardness).

First, we present an efficient algorithm for finding an LO $k$-colouring of a 3-uniform hypergraph that admits an LO $2$-colouring with 
$k=O(\sqrt[3]{n \log \log n / \log n})$.
As mentioned above, there are only a few results on hypergraph colourings with
super-constantly many colours, e.g.~\cite{Krivelevich03:ja,Krivelevich01:ja,Chlamtac08:approx} that deal with hypergraph non-monochromatic colourings.

Second, we establish 
\NP-hardness of finding an LO $k$-colouring of an $r$-uniform hypergraph if an LO $2$-colouring is promised for every constant uniformity $r\geq k + 2$. In fact, we prove a more general result that finding an LO $k$-colouring of an $r$-uniform hypergraph admitting an LO $\ell$-colouring is \NP-hard for every constant $2\leq \ell\leq k$ and $r\geq k - \ell + 4$.
This result is based the algebraic approach to PCSPs and in particular on minions~\cite{BBKO21}.
As a matter of fact, we establish 
relationships of the polymorphism minions of
the LO $2$- vs $k$-colourings
on $r$-uniform hypergraphs,
which may be of independent interest. This gives the  advertised intractability result but
also an impossibility result on certain types of polynomial-time reductions (namely pp-constructions~\cite{BBKO21}) between LO 2- vs $k$-colourings on $r$-uniform hypergraphs with uniformity  
$r\geq k+2$ and $r<k+2$, cf. the discussion in Section~\ref{sec:hardness}.

\paragraph{Related work}
A classic non-monochromatic colouring of a hypergraph requires that no edge should be monochromatic. 
Rainbow colourings, introduced in~\cite{AGH17}, require that every colour should occur at least once in every edge, cf. also~\cite{GL18,ABP20,GS20:rainbow}. LO colourings are somewhat in-between non-monochromatic and rainbow colourings, requiring that the largest colour in every edge is used exactly once. LO colourings --
 as well as a related notion of \emph{conflict-free} colourings, which require that at least one colour in each edge should occur  exactly once --
have been studied under the name of \emph{unique maximum} colourings, both for graphs (where the requirement is on paths)~\cite{Cheilaris11:jda}
and hypergraphs~\cite{Cheilaris13:sidma}.
As far as we know, our results are incomparable with existing results from the rich body of literature on unique maximum colourings. We refer the reader to the papers cited above and the references therein.

\section{Preliminaries}
\label{sec:prelims}

An \emph{$r$-uniform hypergraph} $H$ is a pair $(V, E)$ where $V$ is the set of \emph{vertices} of the hypergraph, and $E \subseteq V^r$ is the set of \emph{edges} of the hypergraph. In our context the order of the vertices in each edge is irrelevant. We will allow vertices to appear multiple times in edges; however, we exclude edges of the form $(v, \ldots, v)$ --- such edges would be impossible in the problems we will consider anyway. We say that two distinct vertices $u, v$ are \emph{neighbours} if they both belong to some edge $e \in E$. Let $N(u)$ be the set of neighbours of $u$. Call a set $S$ an \emph{independent set} of a hypergraph $H$ if and only if no two members of $S$ are neighbours. 

A \emph{linearly ordered} (LO)
$k$-colouring of an $r$-uniform hypergraph $H = (V, E)$ is an assignment $c : V
\to [k]$ of colours from $[k] = \{1, \ldots, k\}$ to the
vertices of $H$ such that, for each edge $(v_1, \ldots, v_r) \in E$, the
sequence $c(v_1), \ldots, c(v_r)$ has a unique maximum. We omit the ``$k$-''
if the number of colours is unimportant.

\begin{example}\label{ex:hyper}
Consider the hypergraph $H = (V, E)$, where $V =[4] = \{1,2,3,4\}$ and $E = \{(1,
2, 3), (1, 2, 4)\}$.
The assignment $c = \{ 1 \mapsto 1, 2 \mapsto 1, 3 \mapsto 2, 4 \mapsto 2 \}$
is an LO 2-colouring, and $c'(x) = x$ is an LO 4-colouring. On the other hand, $c''(x) = 3 - c(x)$ is not an LO colouring at all, since both of the edges have two equal maximal elements when mapped through $c''$.
\end{example}

Finding an LO $k$-colouring, for constant $k\geq 3$, of a $3$-uniform hypergraph
that admits an LO $2$-colouring was studied by Barto et al.~\cite{Barto21:stacs}
in the context of promise constraint satisfaction problems (PCSPs), which we
define next.

\paragraph{Promise CSPs}
Promise CSPs have been introduced in the works of Austrin, Guruswami, and H{\aa}stad~\cite{AGH17} and Brakensiek and Guruswami~\cite{BG21:sicomp}. We follow the notation and terminology of Barto, Bul\'in, Krokhin, and
Opr\v{s}al~\cite{BBKO21}, adapted to structures consisting of a single relation.

An \emph{$r$-ary structure} is a pair $\mathbf{D} = (D, R^\D)$, where $R^\D \subseteq D^r$ and $D$ is finite. We call $D$ the \emph{domain} of the structure, and $R^\D$ the \emph{relation} of the structure. For two $r$-ary structures $\A, \B$, a \emph{homomorphism from $\A$ to $\B$} is a function $h : A \to B$ such that, for each $(a_1, \ldots, a_r) \in R^\A$, we have $(h(a_1), \ldots, h(a_r)) \in R^\B$. This is written $h : \A \to \B$. If we wish to assert only the existence of such a homomorphism, we write $\A \to \B$. 

We now define fixed-template \emph{promise constraint satisfaction problems} (PCSPs). Let $\A$ and $\B$ be two $r$-ary structures with $\A\to\B$; we call $\A$ and $\B$ \emph{templates}. In the \emph{search} version of the PCSP with templates $\A$ and $\B$, denoted by $\PCSP(\A,\B)$, the task is: Given an $r$-ary structure $\I$ with the promise that $\I\to \A$, find a homomorphism from $\I$ to $\B$, which  exists by the composition of promised homomorphism from $\I$ to $\A$ and the assumed homomorphism from $\A$ to $\B$. In the \emph{decision} version of $\PCSP(\A,\B)$, the task is: Given an $r$-ary structure $\I$, output \textsc{yes} if $\I \to \A$, and output \textsc{no} if $\I \not \to \B$. Observe that the decision version can be reduced to the search version: to solve the decision version, run an algorithm for the search version, then check if it gives a correct answer. However, it is not known whether the search version is polynomial-time reducible to the decision version in general. In any case, we will use $\PCSP(\A, \B)$ to mean the decision version when proving hardness, and the search version when showing algorithmic results.

LO colourings can be readily seen as PCSPs. First, observe that an $r$-uniform
hypergraph can be seen as an $r$-ary structure. Second, define an $r$-ary
structure $\LO^r_k$ with domain $[k]$, and whose relation contains
a tuple $(c_1, \ldots, c_r)$ if and only if the sequence $c_1, \ldots, c_r$
has a unique maximum. Then, an LO $k$-colouring of an $r$-uniform hypergraph $\mathbf{H}$
is the same as a homomorphism from $\mathbf{H}$ (viewed as an $r$-ary structure) to
$\LO^r_k$. Thus, the problem of finding an LO $k$-colouring of an $r$-uniform
hypergraph that has an LO $2$-colouring is the same as $\PCSP(\LO_2^r, \LO_k^r)$.

\paragraph{Results}
In Section~\ref{sec:alg}, we study the computational complexity of $\PCSP(\LO_2^3, \LO_{k(n)}^3)$, were
$k(n)$ depends on the input size; here $n$ denotes the number of vertices of the input ($3$-uniform) hypergraph. As in
Example~\ref{ex:hyper}, this is obviously possible for $k(n) = n$. As our first
contribution, we will present in Theorem~\ref{thm:mainAlgo} an efficient algorithm with
$k(n)=O(\sqrt[3]{n \log \log n / \log n})$ colours.

In Section~\ref{sec:hardness}, we study the computational complexity of $\PCSP(\LO_\ell^r, \LO_k^r)$ for constant
uniformity $r\geq 3$ and constant $\ell$ and $k$ with $2\leq \ell \leq k$. We establish
\NP-hardness of $\PCSP(\LO_\ell^r,\LO_k^r)$
for every constant $2\leq \ell\leq k$ and constant uniformity $r\geq k - \ell + 4$, cf. Corollary~\ref{cor:colours}. Our hardness results are based on the algebraic theory of
minions~\cite{BBKO21}, briefly introduced in Section~\ref{sec:minions}. In fact,
we establish relationships between polymorphism minions of
$\PCSP(\LO_2^r,\LO_k^r)$ for all $r\geq 3$,  cf.~Theorem~\ref{thm:minions}.

This is the full version of an ICALP 2022 paper~\cite{NZ22:icalp}, which showed weaker
tractability and intractability results.
Firstly, in~\cite{NZ22:icalp} \NP-hardness of only $\PCSP(\LO_k^r,\LO_{k+1}^r)$ was shown for
$k\geq 2$ and $r\geq 5$.
Secondly, in~\cite{NZ22:icalp} we designed an algorithm that, for a given $3$-uniform hypergraph that admits an LO $2$-colouring, finds an LO $k$-colouring with $k=O(\sqrt{n \log \log n}/\log n)$.\footnote{The previous version of this paper~\cite{NZ22:arxivV1} also contains a weaker result for $4$-uniform hypergraphs.}
Some of the techniques used to establish the tractability result in the current paper are different in character to those used for approximate graph colouring, which was the main technique used in the tractability result in~\cite{NZ22:icalp}.
In particular, the techniques in~\cite{NZ22:icalp} (see also the full version~\cite{NZ22:arxivV1}) built on~\cite{Wigderson83:jacm,Berger90:algoritmica}. We note that we do not know how to apply the SDP-based methods from~\cite{Karger98:jacm} directly as they seem specific to approximate graph colouring and less suited for LO colourings.\footnote{Interestingly, the approximability result for the independent set problem from~\cite{Halldarsson00:JGAA}, which we will use as a black-box in Section~\ref{sec:alg}, does build on~\cite{Karger98:jacm}.}

\section{Algorithmic results}
\label{sec:alg}

In this section we will prove the tractability results advertised above. In particular, we prove the following.

\begin{theorem}\label{thm:mainAlgo}
There is a polynomial-time algorithm that, if given an LO 2-colourable 3-uniform hypergraph with $n$ vertices, finds an LO $O(\sqrt[3]{n \log \log n / \log n})$-colouring.
\end{theorem}

Our algorithm will use the notion of \emph{linear hypergraphs}: we call a 3-uniform hypergraph \emph{linear} if no two edges intersect in 2 or more vertices.\footnote{This needs to hold even for a pair of non-distinct vertices. For instance, we forbid the edges $(1, 1, 2)$ and $(1, 1, 3)$ from existing simultaneously in the hypergraph.}

\begin{proposition}\label{prop:linearisation}
There is a polynomial-time algorithm that, if given an LO 2-colourable 3-uniform hypergraph $H$, constructs an LO 2-colourable linear 3-uniform hypergraph $H'$ with no more vertices than $H$ such that, if given an LO $k$-colouring of $H'$, one can compute in polynomial time an LO $k$-colouring of $H$.
\end{proposition}
\begin{proof}
Suppose $H$ is not linear already. Then it has edges $(x, y, a)$ and $(x, y, b)$. In all LO 2-colourings of $H$, $a$ and $b$ must be assigned the same colours. Therefore it is safe to merge $a$ and $b$. Repeat this procedure until $H$ is linear. To find an LO $k$-colouring of $H$ given one of $H'$, simply undo the merges.
\end{proof}

In order to more elegantly express the algorithm we now propose, we will adapt the notion of ``progress'', introduced by Blum~\cite{Blum94}, to approximate LO colouring.

\begin{definition}
To make progress towards an LO $O(f(n))$-colouring of a 3-uniform hypergraph consists of one of the following:
\begin{description}
    \item[Type 1] Finding an independent set of size $\Omega(n / f(n))$.
    \item[Type 2] Finding a set of $\Omega(n / f(n))$ vertices that intersects each edge in 0 or 2 vertices.
\end{description}
\end{definition}

That this notion of progress is useful is shown in the following proposition.

\begin{proposition}
Suppose there is a polynomial-time algorithm that, if given a \emph{linear} LO 2-colourable 3-uniform hypergraph, makes progress towards an LO $O(f(n))$-colouring, where $f(n) = \Theta(n^\alpha (\log n)^\beta (\log \log n)^\gamma)$ and $\alpha > 0$. Then, there is a polynomial-time algorithm that finds an LO $O(f(n))$-colouring of an LO 2-colourable 3-uniform hypergraph.
\end{proposition}

\begin{proof}
We will describe a recursive procedure that achieves our goal. 
Due to Proposition~\ref{prop:linearisation}, we can assume that our input hypergraph is linear.

Now, at each recursive call, apply our progress-making algorithm. If we make type 1 progress, then colour the set of vertices with a \emph{large colour}; if we make type 2 progress, then colour the set of vertices with a \emph{small colour}. In any case, remove all edges incident to the coloured set.

First, note that this algorithm always produces a valid LO colouring. It suffices to note that for type 1 progress, any edge that was not coloured recursively intersects the maximum colour in one vertex; whereas for type 2 progress, any edge that was not coloured recursively intersects the minimum colour in two vertices. All such edges are thus coloured correctly.

To see why we use only $O(f(n))$ colours, consider the number of iterations needed to reduce the number of uncoloured vertices in the hypergraph by half --- it is not difficult to see, for our choice of $f$, that it takes at most $O(n / (n / f(n))) = O(f(n))$ iterations. Applying the Master Theorem \cite{CLRS}, noting that $f(n) = \Omega(n^{\alpha'})$ for some $\alpha' > 0$, it follows that there are $O(f(n))$ iterations overall. (We can take any $\alpha'$ such that $0 < \alpha' < \alpha$.)
This bound on the number of iterations also immediately gives us polynomial run time.
\end{proof}

With this framework in hand, we are ready to prove our result. We will essentially split into two cases: one for sparse hypergraphs, and one for dense hypergraphs.

\begin{proposition}\label{prop:smallDeg}
Fix $\Delta = \Theta(n^
\alpha (\log n)^\beta (\log \log n)^\gamma)$, with $\alpha > 0$. There is a polynomial-time algorithm that makes progress towards an LO $O(\Delta \log \log n / \log n)$-colouring of an LO 2-colourable 3-uniform linear hypergraph if $|E| = O(n \Delta)$.
\end{proposition}

Fix $H$ to be our linear hypergraph. We will prove Proposition~\ref{prop:smallDeg}  with a series of lemmata.
\begin{lemma}
$H$ has $\Omega(n)$ vertices with degree $O(\Delta)$.
\end{lemma}
\begin{proof}
Observe that the average degree of $H$ is $O(\Delta)$. Applying Markov's inequality to a randomly chosen vertex immediately gives us our result.
\end{proof}

Let $V'$ be the set of vertices of $H$ with degree $O(\Delta)$, and let $E'$ be the set of edges induced by $V'$.

\begin{lemma}
An independent set of size $|E'| / \Delta$ exists within $V'$.
\end{lemma}
\begin{proof}
Each edge in $E'$ must contain at least one vertex coloured with 1 in an LO 2-colouring; each such vertex is included in at most $\Delta$ edges. Thus there exists an independent set of size $|E'| / \Delta$.
\end{proof}

We will use the following algorithm as a black box.

\begin{lemma}[\cite{Halldarsson00:JGAA}]
There is a polynomial-time algorithm that, if given a graph with an independent set of size $s$ and average degree $d$, finds an independent set of size at least $\Omega(s \log d / d \log \log d)$.
\end{lemma}

\begin{corollary}\label{cor:IS}
There is a polynomial-time algorithm that, if given a graph with an independent set of size $s$ and average degree $d$, finds an independent set of size at least 
  \[\max(\Omega(s \log d / d \log \log d), \Omega(n / d)).\]
\end{corollary}
\begin{proof}
As established above for the more general case of hypergraphs, since the input graph has $\Theta(nd)$ vertices, there must be $\Omega(n)$ edges with degree $O(d)$. Thus, by greedily selecting an independent set on these vertices, we get the lower bound of $\Omega(n / d)$. Take the larger of this independent set and the one generated by the previous algorithm to get the desired result.
\end{proof}

\begin{proof}[Proof of Proposition~\ref{prop:smallDeg}]
Apply the algorithm from Corollary~\ref{cor:IS} to the primal graph of $H' = (V', E')$. (The primal graph of a hypergraph replaces each hyperedge with a clique.) This graph has an independent set of size $|E'| / \Delta$ and average degree $d = O(|E'| / |V'|) = O(|E'|/n)$. Therefore, the algorithm will find an independent set of size at least equal to the maximum of $\Omega((|E'| / \Delta) \log d / (|E'| / n) \log \log d) = \Omega(n \log d / \Delta \log \log n)$ and $\Omega(|V'| / d) = \Omega(n / d)$. If $d = o(\Delta \log \log n / \log n)$ then we get the independent set size we want from $\Omega(n / d)$. On the other hand, if $d = \Omega(\Delta \log \log n / \log n) = \Omega(n^{\alpha - \epsilon})$ for some $\epsilon < \alpha$, then we get the independent set size we want from $\Omega(n \log d / \Delta \log \log d)$ --- the assumption on the size of $d$ is necessary to make $\log d / \log \log d = \Omega(\log n / \log \log n)$.
\end{proof}

\begin{proposition}\label{prop:bigDeg}
There is a polynomial-time algorithm that makes progress towards an LO $O(\sqrt{n / \Delta})$-colouring of an LO 2-colourable 3-uniform linear hypergraph if $|E| = \Omega(n \Delta)$.
\end{proposition}

Fix $H$ to be a 3-uniform linear hypergraph that admits an LO 2-colouring. The following two lemmata prove Proposition~\ref{prop:bigDeg}.

\begin{lemma}
There exists an LO 2-colouring of $H$ with $\Omega(\sqrt{n\Delta})$ vertices set to 1.
\end{lemma}
\begin{proof}
Consider any LO 2-colouring of $H$ (one must exist by assumption). Consider the set $S$ of vertices set to 1. Each edge intersects $S$ in exactly two vertices; furthermore, every pair of vertices is included in at most one edge by linearity. Thus, by connecting these pairs of vertices, we find that we can construct a simple graph on $S$ with $|E|$ edges. Thus $|S|^2 \geq |E| = \Omega(n\Delta)$, from which it follows that $|S| = \Omega(\sqrt{n \Delta})$.
\end{proof}

\begin{definition}
For any 3-uniform hypergraph $H$, let $\mathcal{E}(H)$ denote a set of linear equations mod~2. These equations have one variable $x_v$ for each vertex $v$ of $H$, and one equation for each edge of $H$. The equation for edge $(a, b, c)$ is $x_{a} + x_b + x_c \equiv 0 \bmod 2$.
\end{definition}

\begin{lemma}[{\cite[Theorem~2.12,~(2)]{DBLP:journals/KSTW00}}]\label{lem:MaxOnes}
The problem of finding a solution to a set of linear equations mod~2 that maximises the number of variables set to 1 can be approximated to within some constant factor in polynomial time.
\end{lemma}

\begin{proof}[Proof of Proposition \ref{prop:bigDeg}]
Consider $\mathcal{E}(H)$. Observe that a set of vertices that intersects each edge in an even number of vertices is equivalent to a solution to $\mathcal{E}(H)$. Furthermore, the size of the set is the same as the number of variables set to one in this solution; it is therefore sufficient to show that we can find a solution to this system of linear equations mod~2 that sets $\Omega(\sqrt{n \Delta})$ variables to one.

Now, consider the LO 2-colouring of $H$ that sets $\Omega(\sqrt{n\Delta})$ variables to 1. By recolouring 2 into 0, this gives a solution to $\mathcal{E}(H)$ that sets at least $\Omega(\sqrt{n \Delta})$ variables to 1. Finally, using the algorithm from Lemma~\ref{lem:MaxOnes}, we can find a solution that is at most a constant factor away from this one --- i.e.~with at least $\Omega(\sqrt{n \Delta})$ variables set to 1. This concludes the proof.
\end{proof}

\begin{proof}[Proof of Theorem~\ref{thm:mainAlgo}]
Set $\Delta = n^{1/3} (\log n)^{2 / 3} (\log \log n)^{-2/3}$ so that $\sqrt{n / \Delta} = \Delta \log \log n / \log n$
and combine Proposition~\ref{prop:smallDeg} and Proposition~\ref{prop:bigDeg}.
Thus, we can always make progress towards an LO $O(\sqrt[3]{n \log \log n / \log n})$-colouring in polynomial time, if given an LO 2-colourable 3-uniform linear hypergraph.
\end{proof}

\section{Algebraic theory of fixed-template promise CSPs}
\label{sec:minions}

We recount the algebraic theory of fixed-template PCSPs developed
in~\cite{BBKO21} and specialised to templates with a single relation (of arity
$r$).

The \emph{$p$-the power} of an $r$-ary template $\A=(A,R^{\A})$ is a template $\A^p = (A^p, R^{\A^p})$ where
\[
R^{\A^p} = \{ ( ( a_1^1, \ldots, a_1^p), \ldots, (a_r^1, \ldots, a_r^p ) ) \mid (a_1^1, \ldots, a_r^1) \in R^\A, \ldots, (a_1^p, \ldots, a_r^p) \in R^\A \}.
\]
In other words, a tuple of $R^{\A^p}$ contains $r$ vectors of $p$ elements of $A$, such that if these are written as a matrix with $r$ rows and $p$ columns, each column is a member of $R^\A$. For two $r$-ary templates $\A, \B$, a $p$-ary polymorphism from $\A$ to $\B$ is a homomorphism $f : \A^p \to \B$. We let $\ar(f)$ denote the arity of a polymorphism.

\begin{example}
Consider the binary template $\A = ([2], R^\A)$, where $R^\A$ is the binary disequality relation $\neq$ (restricted to $[2]^2$). The power $\A^5$ has domain $[2]^5$ and relation $\{ (\mathbf{a}, \mathbf{b}) \mid \mathbf{a}, \mathbf{b} \in [2]^5, a_i \neq b_i, i = 1, \ldots, 5 \}$. This relation is constructed as follows: $(\mathbf{a}, \mathbf{b})$ belongs to the relation if and only if every column of a matrix with 5 columns and 2 rows constructed out of $\mathbf{a}, \mathbf{b}$ satisfies $\neq$. The matrix is the following one:
\[
\begin{pmatrix}
a_1 & a_2 & a_3 & a_4 & a_5 \\
b_1 & b_2 & b_3 & b_4 & b_5 \\
\end{pmatrix}.
\]
Thus, for each column to satisfy $\neq$, we must have $a_i \neq b_i$ for $i = 1 , \ldots, 5$, as indicated above.

Now, consider a quinary polymorphism $f : \A^5 \to \A$. This is a function $f : [2]^5 \to [2]$ that satisfies the following property: if given a matrix with 2 rows and 5 columns, such that each column is a member of $R^\A$, then by applying $f$ to the rows of this matrix we also get a member of $R^\A$. For instance, for the matrix
\[
\begin{pmatrix}
1 & 2 & 2 & 1 & 1 \\
2 & 1 & 1 & 2 & 2
\end{pmatrix},
\]
we deduce that the pair $(f(1, 2, 2, 1, 1), f(2, 1, 1, 2, 2)) \in R^\A$ i.e.~$f(1, 2, 2, 1, 1) \neq f(2, 1, 1, 2, 2)$. One such polymorphism is given by selecting the values of $f$ from $[2]$ such that $f(x_1,\ldots, x_5) \equiv \sum_{i=1}^5 x_i \pmod{2}$.
\end{example}

The real power of this theory comes from \emph{minions}.\footnote{In category-theoretic terms, a minion is a functor from the skeleton of the category of finite sets to the category of sets. The objects of the first category are sets $[p]$ for $p \in \nat$, and the arrows are functions between them. The functor equivalent to a minion $\M$ takes $[p]$ to $\M^{(p)}$, and $\pi : [p] \to [q]$ to $f \mapsto f_\pi$.}
A \emph{minion} $\M$ is a sequence of sets $\M^{(0)}, \M^{(1)}, \ldots$, equipped with (so-called minor) operations, one for each $\pi : [p] \to [q]$; given $f \in \M^{(p)}$ the minor operation yields $f_\pi \in \M^{(q)}$. The operation must satisfy the following conditions:
\begin{itemize}
\item For $f \in \M^{(p)}$, if $\mathit{id} : [p] \to [p]$ is the identity on $[p]$, then $f_\mathit{id} = f$.
\item For $f \in \M^{(p)}$, $\pi : [p] \to [q]$ and $\sigma : [q] \to [t]$, we have $f_{\sigma \circ \pi} = (f_\pi)_\sigma$.
\end{itemize}
An important class of minions is class of \emph{polymorphism minions}. The polymorphism minion $\M = \Pol(\A, \B)$ for two templates $\A, \B$ with the same arity is a minion where $\M^{(p)}$ is the set of $p$-ary polymorphisms from $\A$ to $\B$, and where, for $f : A^p \to B$ and $\pi : [p] \to [q]$, $f_\pi$ is given by
$f_\pi(x_1, \ldots, x_q) = f(x_{\pi(1)}, \ldots, x_{\pi(p)})$.\footnote{Put differently, minors in polymorphism minions permute variables, identify variables, and introduce dummy variables.}
It is not difficult to check that, if $f : \A^p \to \B$ and $\pi : [p] \to [q]$, then $f_\pi : \A^q \to \B$, as required.

In order to be able to relate polymorphism minions with the complexity of PCSPs, we use \emph{minion homomorphisms}.\footnote{In category-theoretic terms, a minion homomorphism is just a natural transformation.}
A \emph{minion homomorphism from $\M$ to $\N$} is a mapping $\xi$ that takes each $\M^{(p)}$ to $\N^{(p)}$ and that satisfies the following condition: for any $\pi : [p] \to [q]$ and $f \in \M^{(p)}$, $\xi(f)_\pi = \xi(f_\pi)$. The following theorem links minion homomorphisms to PCSPs in the sense that minion homomorphisms capture precisely a certain type of polynomial-time reductions, know as primitive-positive  constructions.\footnote{Primitive-positive constructions (or pp-constructions, for short) capture so-called ``gadget reductions'', cf.~\cite[Section~3]{BKW17}.} 

\begin{theorem}[{\cite[Theorems~3.1~and~4.12]{BBKO21}}]\label{thm:gadgets}
For $r$-ary templates $\A, \B$ and $r'$-ary templates $\A', \B'$, a primitive-positive construction-based polynomial-time reduction from $\PCSP(\A', \B')$ to $\PCSP(\A, \B)$ exists if and only if $\Pol(\A, \B) \to \Pol(\A', \B')$.
\end{theorem}

In particular, a minion homomorphism between polymorphism minions implies a polynomial-time reduction (in the other direction).
Unfortunately, it is usually a complex task to explicitly construct minion homomorphisms. An auxiliary construction called the \emph{free structure} allows us to construct them more easily. For an arbitrary minion $\M$ and an $r$-ary template $\A = (A, R^\A)$, the \emph{free structure} $\F = \F_\M(\A)$ of $\M$ generated by $\A$ is an $r$-ary template whose domain is $F = \M^{(|A|)}$. To construct its relation $R^\F$, first identify $A$ with $[n]$ for $n = |A|$, and then enumerate the tuples of $R^\A$ as vectors $\mathbf{r}^1, \ldots, \mathbf{r}^k$, where $k = |R^\A|$. Construct functions $\pi_1, \ldots, \pi_r : [k] \to [n]$ where $\pi_i(j) = \mathbf{r}^j_i$. (If we were to arrange $\mathbf{r}^1, \ldots, \mathbf{r}^k$ as columns of a matrix with $r$ rows and $k$ columns, then $\pi_i(1), \ldots, \pi_i(k)$ is the $i$-th row of the matrix.) Now, the tuple $(f_1, \ldots, f_r)$, where $f_1, \ldots, f_r \in \M^{(n)}$, belongs to $R^\F$ if and only if for some $f \in \M^{(k)}$ we have $f_i = f_{\pi_i}$.

\newcommand\cx[1]{\colorbox{white}{\strut #1}}
\newcommand\cy[1]{\colorbox{pink}{\strut #1}}
\newcommand\cz[1]{\colorbox{yellow}{\strut #1}}

\renewcommand\cx[1]{#1}
\renewcommand\cy[1]{#1}
\renewcommand\cz[1]{#1}

\begin{example}
Consider some polymorphism minion $\M$ and the ternary template $\LO_3^3$.
We will construct $\F = \F_\M(\LO_3^3)$. The domain is $\M^{(3)}$. To construct the relation of $\F$, we first arrange the 15 tuples of $R^{\LO_3^3}$ into a matrix with 3 rows and 15 columns:
\[
\begin{pmatrix}
\cy{2} & \cx{1} & \cx{1} & \cz{3} & \cx{1} & \cx{1} & \cz{3} & \cy{2} & \cz{3} & \cx{1} & \cy{2} & \cx{1} & \cz{3} & \cy{2} & \cy{2} \\
\cx{1} & \cy{2} & \cx{1} & \cx{1} & \cz{3} & \cx{1} & \cy{2} & \cz{3} & \cx{1} & \cz{3} & \cx{1} & \cy{2} & \cy{2} & \cz{3} & \cy{2} \\
\cx{1} & \cx{1} & \cy{2} & \cx{1} & \cx{1} & \cz{3} & \cx{1} & \cx{1} & \cy{2} & \cy{2} & \cz{3} & \cz{3} & \cy{2} & \cy{2} & \cz{3}
\end{pmatrix}
.
\]
Row $i$ of this matrix can be seen as a function $\pi_i : [15] \to [3]$. Now the relation $R^\F$ contains precisely the tuples $(f_{\pi_1}, f_{\pi_2}, f_{\pi_3})$ for all $f \in \M^{(15)}$. Substituting the definition for $f_{\pi_i}$, we find that these polymorphisms $f_{\pi_1}, f_{\pi_2}, f_{\pi_3}$ are:
\begin{align*}
(x, y, z) \mapsto &f\left(\cy{y}, \cx{x}, \cx{x}, \cz{z}, \cx{x}, \cx{x}, \cz{z}, \cy{y}, \cz{z}, \cx{x}, \cy{y}, \cx{x}, \cz{z}, \cy{y}, \cy{y}\right) \\
(x, y, z) \mapsto &f\left(\cx{x}, \cy{y}, \cx{x}, \cx{x}, \cz{z}, \cx{x}, \cy{y}, \cz{z}, \cx{x}, \cz{z}, \cx{x}, \cy{y}, \cy{y}, \cz{z}, \cy{y}\right) \\
(x, y, z) \mapsto &f\left(\cx{x}, \cx{x}, \cy{y}, \cx{x}, \cx{x}, \cz{z}, \cx{x}, \cx{x}, \cy{y}, \cy{y}, \cz{z}, \cz{z}, \cy{y}, \cy{y}, \cz{z}\right)
\end{align*}
Observe that the matrix and the arguments of $f$ are actually arranged in the same configuration, with 1 replaced by $x$, 2 by $y$ and 3 by $z$.
\end{example}

The following theorem connects minion homomorphisms and free structures.

\begin{theorem}[{\cite[Lemma~4.4]{BBKO21}}]\label{thm:adjunction}
If $\M$ is a minion and $\A, \B$ are $r$-ary templates, the homomorphisms $h : \F_\M(\A) \to \B$ are in a (natural)
1-to-1 correspondence to the minion homomorphisms $\xi : \M \to \Pol(\A, \B)$.\footnote{In category-theoretic terms, $\F_-(\A)$ and $\Pol(\A, -)$ are functors between (in opposite directions) the category of $r$-ary templates and the category of minions, and $\F_-(\A) \dashv \Pol(\A, -)$.}
As a consequence, $\F_\M(\A) \to \B$ if and only if $\M \to \Pol(\A, \B)$.
\end{theorem}

\section{Hardness results}
\label{sec:hardness}

In this section we will investigate the hardness of $\PCSP(\LO_\ell^r, \LO_k^r)$. 

First, we will establish that $\PCSP(\LO_2^r,\LO_k^r)$ is \NP-hard for each $k$, for \emph{some} (large but constant) $r$ (cf. Theorem~\ref{thm:sourceOfHardness}). We then use minion homomorphisms to show that $\PCSP(\LO_\ell^r, \LO_k^r)$ is \NP-hard for $r \geq k - \ell + 4$ (cf. Corollary~\ref{cor:colours}).

Second, we will show that $\PCSP(\LO_2^r,\LO_k^r)$ with $r\geq k+2$ \emph{cannot} be reduced to $\PCSP(\LO_2^r,\LO_k^r)$ with $r<k+2$
using primitive-positive constructions (i.e.~gadget reductions~\cite{BBKO21}). Thus in particular, it is not possible to prove hardness of $\PCSP(\LO_2^3,\LO_3^3)$ and $\PCSP(\LO_2^4,\LO_3^4)$ via gadget reductions from $\PCSP(\LO_2^5,\LO_3^5)$. More generally, if $\PCSP(\LO_2^{r'},\LO_k^{r'})$ is proved NP-hard for \emph{some} $r'\geq k+2$ then this will imply NP-hardness of $\PCSP(\LO_2^r,\LO_k^r)$ for \emph{every} $r\geq k+2$ but no gadget reductions would imply NP-hardness of $\PCSP(\LO_2^r,\LO_k^r)$ for $2\leq r<k+2$ (cf. Theorem~\ref{thm:minions}).

\subsection{Source of hardness}

In this subsection we will show that, for each $k \geq 2$, $\PCSP(\LO_2^r, \LO_k^r)$ is \NP-hard for some large $r$.

\begin{theorem}\label{thm:sourceOfHardness}
For any $k \geq 2$, there exists some large $r$ such that $\PCSP(\LO_2^r, \LO_k^r)$ is \NP-hard.
\end{theorem}

To do this, we will construct a hardness condition and apply it to this problem. The following hardness condition is almost identical to the one from~\cite{SetSat} (later reformulated in~\cite{babyPCP}). It is somewhat more general, in that the arity of the polymorphisms are bounded --- we do not need to prove something for polymorphisms of arbitrarily large arity.

For the purposes of this section, a \emph{chain of minors} is a sequence of polymorphisms with minors between them: $f_0 \xrightarrow{\pi_{0,1}} \ldots \xrightarrow{\pi_{k-1,k}} f_k$. We let $\pi_{i,j}$ denote the composition of the minors between $f_i$ and $f_j$ i.e.~$\pi_{i,j} = \pi_{j-1,j} \circ \ldots \circ \pi_{i,i+1}$. Thus $f_i \xrightarrow{\pi_{i,j}} f_j$.

\begin{theorem}\label{thm:condition}
Suppose $\M = \Pol(\A, \B)$. Fix constants $\ell, k\in \mathbb{N}$. There exists a constant $m = m(\ell,k)$ such that the following holds. Suppose that for each polymorphism $f : \A^n \to \B$ where $n \leq m$ we assign a set $I(f) \subseteq [n]$ of size at most $k$, such that for each chain of minors $f_0 \xrightarrow{\pi_{0,1}} \ldots \xrightarrow{\pi_{\ell-1,\ell}} f_\ell$ containing polymorphisms of arity at most $m$, there exist $i, j$ such that $\pi_{i,j}(I(f_i)) \cap I(f_j) \neq \emptyset$. Then $\pmc_\M(m)$ is \NP-hard, and furthermore $\PCSP(\A, \B)$ is \NP-hard.
\end{theorem}
We do not need the definition of the ``promise satisfaction of a bipartite minor condition'' problem from~\cite{BBKO21}, denoted by $\pmc_\M(m)$ in the statement, as we never use it; it is only included to match the result from~\cite{SetSat}.
\begin{proof}
Identical to the proof in \cite[Corollary~4.2]{SetSat}, but noting that $I$ (called $\mathrm{sel}$ in \cite{SetSat}) is only ever applied on polymorphisms with arity at most $m$, and that the selection of $m$ depends only on $\ell$ and $k$.
\end{proof}

We define $\M_{\ell, k}^r =
\Pol(\LO_\ell^r, \LO_k^r)$. In order to apply this condition, we will have to better understand $\M_{\ell, k}^r$ from a combinatorial point of view --- the following remark gives a useful way to see the polymorphisms in $\M_{\ell, k}^r$.

\begin{remark}\label{rem:reinterpret}
Observe that an $(r')$-ary polymorphism $f \in (\M_{2,k}^r)^{(r')}$ is a function from $[2]^{r'}$ to $[k]$; if it is applied to the rows of an $r \times r'$ matrix whose columns are tuples of the relation of $\LO_2^r$ (i.e.~they contain one 2 and otherwise are 1), then the resulting values contain a unique maximum. Similarly to Barto et.~al.~\cite{Barto21:stacs}, we view $f$ as a function from the powerset of $[r']$ to $[k]$ indicating the coordinates of $2$s. (For example, the input tuple $(1, 2, 1, 2)$ is seen as equivalent to the input set $\{ 2, 4 \}$.) Under this view, $f$ is a polymorphism if and only if, for any partition $A_1, \ldots, A_r$ of $[r']$, the sequence $f(A_1), \ldots, f(A_r)$ has a unique maximum element. (Observe that each part $A_i$ corresponds to a row in the matrix mentioned earlier. Put differently, part $A_i$ corresponds to columns in which the $i$-th row contains a (unique in its column) 2.)
\end{remark}

\begin{lemma}\label{lem:partitions}
Suppose $f : 2^{[n]} \to [k]$ is a function such that, for each partition $A_1, \ldots, A_p$, $p \in [n]$, of $[n]$, $f(A_1), \ldots, f(A_p)$ has a unique maximum.

Fix a partition $A_1, \ldots, A_p$ of $[n]$. Suppose that the unique maximum of $f(A_1), \ldots, f(A_p)$ is equal to the unique maximum of $f \{1\}, \ldots, f\{n\}$. Then, if the unique maximum of the first is $f(A_i)$ and the unique maximum of the second is $f\{j\}$, it follows that $j \in A_i$.
\end{lemma}
\begin{proof}
Suppose not. Thus suppose $j \not \in A_i$. Consider the partition $A_i, \{1\}, \ldots,\{j\}, \ldots, \{n\}$, where all singletons included in $A_i$ have been removed. This partition must have a unique maximum. Since $f(A_i) = f \{ j\}$, it must be larger than $f \{j \}$ and $f(A_i)$. But this is impossible, since $f \{j \}$ is the maximum of $f\{1\}, \ldots, f\{n\}$.
\end{proof}

\begin{proof}[Proof of Theorem~\ref{thm:sourceOfHardness}]
Fix $m = m(1, k)$ and $r = m + 2$. Consider any $f \in (\M_{2, k}^r)^{(n)}$ for $n \leq m$. For any partition $A_1, \ldots, A_p$ of $[n]$ where $p \in [n]$, observe that due to the partition $A_1, \ldots, A_p, \emptyset, \ldots, \emptyset$, where $\emptyset$ is added $r - p \geq m + 2 - m = 2$ times, it follows that $f(A_1), \ldots, f(A_p)$ has a unique maximum. This means that we can apply Lemma~\ref{lem:partitions} to any such $f$.

For any $f \in (\M_{2,k}^r)^{(n)}$ for $n \leq m$ it follows that $f\{1\}, \ldots, f\{n\}$ has a unique maximum. If it is given by $f\{i\}$, then set $I(f) = \{i\}$. We will now show that this selection $I$ satisfies the condition of Theorem~\ref{thm:condition}, and thus $\PCSP(\LO_2^r, \LO_k^r)$ is \NP-hard.

Consider any chain of minors $f_0 \xrightarrow{\pi_{0,1}} \ldots \xrightarrow{\pi_{k-1, k}} f_k$. By the pigeonhole principle applied to $f_0(I(f_0)), \ldots, f_k(I(f_k))$, for some $f = f_i \xrightarrow{\pi = \pi_{i,j}} f_j = g$ we have $f(I(f)) = g(I(g)) = c \in [k]$. Supposing that $p = \ar(f), q = \ar(g)$, $p, q \leq m$, it follows that the unique maximum value in $f \{1\}, \ldots, f\{p\}$ and in $g\{1\}, \ldots, g\{q\}$ are both $c$. Suppose the first is given by $f\{i\}$ and the second is given by $g\{j\}$. Thus $I(f) = \{i\}$ and $I(g) = \{j\}$.

Observe that, by definition, $g\{x\} = f(\pi^{-1}(x))$. It follows that the unique maximum of 
  \[f(\pi^{-1}(1)), \ldots, f(\pi^{-1}(q))\]
  is given by $f(\pi^{-1}(j))$. Since additionally $f\{i\} = f(\pi^{-1}(j))$ and $f\{i\}$ is the unique maximum of $ f\{1\}, \ldots, f\{i\}$, by Lemma~\ref{lem:partitions} it follows that $i \in \pi^{-1}(j)$ or equivalently $\pi(i) = j$. This implies that $\pi(I(f)) \cap I(g) = \{j\} \neq \emptyset$. Thus by Theorem~\ref{thm:condition} it follows that $\PCSP(\LO_2^r, \LO_k^r)$ is \NP-hard.
\end{proof}

\subsection{Minion homomorphisms}

How can we now leverage this basic hardness result to other values of $r$? We
will use chains of minion homomorphisms to do this. Our main result in this section is the following.

\begin{theorem}\label{thm:minions}
For each $k \geq 3$, we have that $\M_{2,k}^{k+2} \leftrightarrows \M_{2,k}^{k+3} \leftrightarrows \ldots$. Furthermore for each $2 \leq r < k+2$, $\M_{2, k}^{k+2} \to \M_{2, k}^r$, yet $\M_{2, k}^r \not \to \M_{2, k}^{k+2}$.
\end{theorem}
Theorem~\ref{thm:minions} is illustrated in Figure~\ref{fig:minions0}. 
The more complicated Figure~\ref{fig:minions} illustrates minion relationships
from Theorem~\ref{thm:minions} and those implied by
Theorems~\ref{thm:minion-2},\ref{thm:impossible1}, and~\ref{thm:impossible2}.

\begin{remark}
For $k=3$, we know the precise relationship of all minions in Theorem~\ref{thm:minions}: $\M_{2,3}^3\not\to\M_{2,3}^4$ and $\M_{2,3}^4\not\to\M_{2,3}^3$; i.e, the (in this case) two minions in the ``left fan'' in Figure~\ref{fig:minions} are incomparable. This follows from Theorem~\ref{thm:impossible1} and Theorem~\ref{thm:impossible2}.
\end{remark}
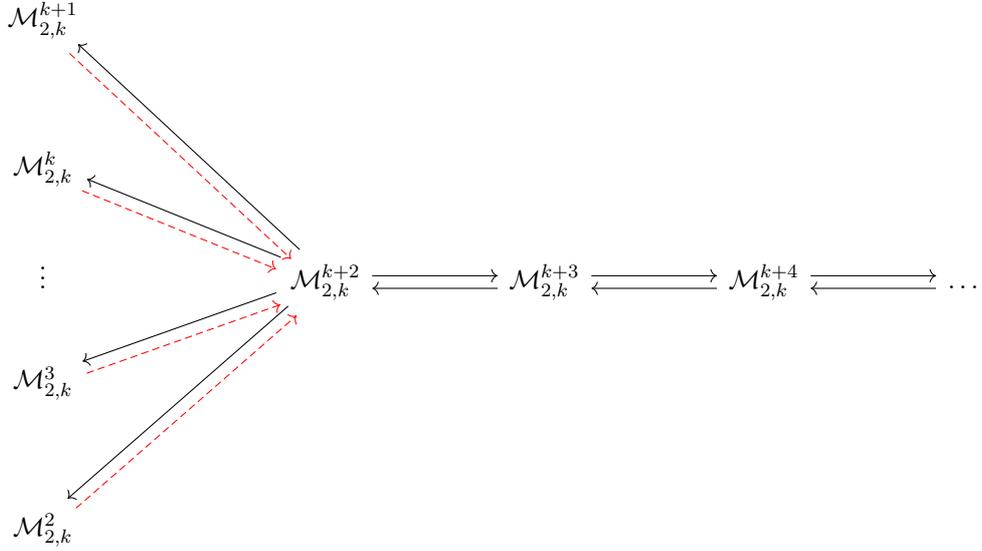
\begin{figure}[htb]
    \centering
\[\adjustbox{scale=0.9,center}{
  \begin{tikzcd}
	{\M_{2,k}^{k+1}} \\
	\\
	{\M_{2,k}^{k}} \\
	\vdots &&& {\M_{2,k}^{k+2}} && {\M_{2,k}^{k+3}} && {\M_{2,k}^{k+4}} && \ldots \\
	{\M_{2,k}^{3}} \\
	\\
	{\M_{2,k}^{2}}
	\arrow[shift right=1, from=4-4, to=1-1]
	\arrow[shift right=1, from=4-4, to=3-1]
	\arrow[shift right=1, from=4-4, to=5-1]
	\arrow[shift left=1, from=4-4, to=4-6]
	\arrow[shift left=1, from=4-6, to=4-8]
	\arrow[shift left=1, from=4-6, to=4-4]
	\arrow[shift left=1, from=4-8, to=4-6]
	\arrow[shift right=1, 
	dashed, color=red, from=1-1, to=4-4]
	\arrow[shift right=1,
	dashed, color=red, from=3-1, to=4-4]
	\arrow[shift right=1, 
	dashed, color=red, from=5-1, to=4-4]
	\arrow[shift right=1, from=4-4, to=7-1]
	\arrow[shift right=1,
	dashed, color=red, from=7-1, to=4-4]
	\arrow[shift left=1, from=4-8, to=4-10]
	\arrow[shift left=1, from=4-10, to=4-8]
\end{tikzcd}}
\]
    \caption{Minions from Theorem~\ref{thm:minions}. Solid black arrows indicate the existence of minion homomorphisms, whereas red dashed arrows indicate the non-existence of minion homomorphisms.}
    \label{fig:minions0}
\end{figure}
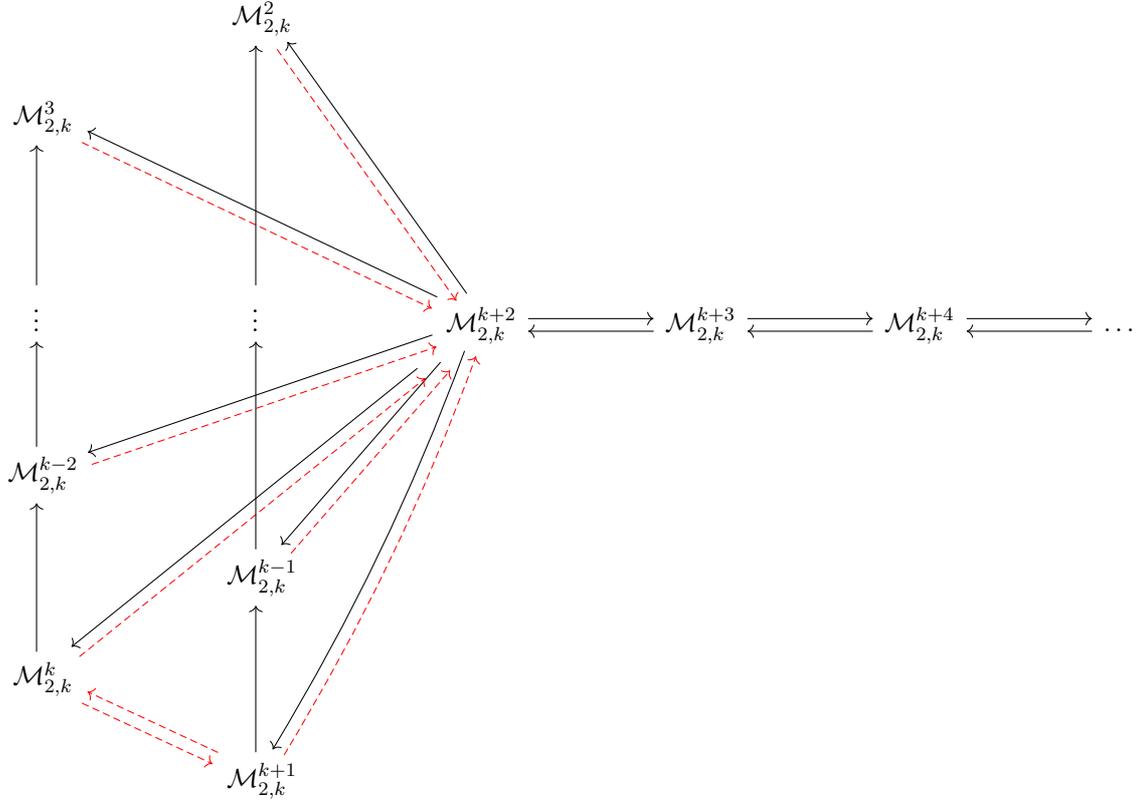
\begin{figure}[hbt]
    \centering
\[\adjustbox{scale=0.9,center}{
\begin{tikzcd}
	&& {\mathcal{M}_{2,k}^{2}} \\
	{\mathcal{M}_{2,k}^{3}} \\
	\\
	\\
	{\vdots\,\,\,} && {\vdots\,\,\,} && {\mathcal{M}_{2,k}^{k+2}} && {\mathcal{M}_{2,k}^{k+3}} && {\mathcal{M}_{2,k}^{k+4}} && \ldots \\
	\\
	{\mathcal{M}_{2,k}^{k-2}} \\
	&& {\mathcal{M}_{2,k}^{k-1}} \\
	{\mathcal{M}_{2, k}^k} \\
	&& {\mathcal{M}_{2,k}^{k+1}}
	\arrow[shift right=1, from=5-5, to=1-3]
	\arrow[shift right=1, from=5-5, to=2-1]
	\arrow[shift right=1, from=5-5, to=7-1]
	\arrow[shift left=1, from=5-5, to=5-7]
	\arrow[shift left=1, from=5-7, to=5-9]
	\arrow[shift left=1, from=5-7, to=5-5]
	\arrow[shift left=1, from=5-9, to=5-7]
	\arrow[shift right=1, dashed, color=red, from=1-3, to=5-5]
	\arrow[shift right=1, dashed, color=red, from=2-1, to=5-5]
	\arrow[shift right=1, dashed, color=red, from=7-1, to=5-5]
	\arrow[shift right=1, shorten <=7pt, from=5-5, to=8-3]
	\arrow[shift right=1, shorten >=7pt, dashed, color=red, from=8-3, to=5-5]
	\arrow[shift left=1, from=5-9, to=5-11]
	\arrow[shift left=1, from=5-11, to=5-9]
	\arrow[shift right=1, curve={height=6pt}, 
	dashed, color=red, from=10-3, to=5-5]
	\arrow[shift right=1, curve={height=-6pt}, 
	from=5-5, to=10-3]
	\arrow[shift left=1, from=10-3, to=8-3]
	\arrow[shift left=1, from=9-1, to=7-1]
	\arrow[shift right=1, shorten <=13pt, from=5-5, to=9-1]
	\arrow[shift right=1, shorten >=13pt, dashed, color=red, from=9-1, to=5-5]
	\arrow[shift left=1, from=5-1, to=2-1]
	\arrow[shift left=1, from=7-1, to=5-1]
	\arrow[shift left=1, from=8-3, to=5-3]
	\arrow[shift left=1, from=5-3, to=1-3]
	\arrow[shift right=1, dashed, color=red, from=9-1, to=10-3]
	\arrow[shift right=1, dashed, color=red, from=10-3, to=9-1]
\end{tikzcd}}\]
    \caption{
    Minions from Theorems~\ref{thm:minions}, \ref{thm:minion-2},
    \ref{thm:impossible1}, and~\ref{thm:impossible2}. Solid black arrows indicate the existence of minion homomorphisms, whereas red dashed arrows indicate the non-existence of minion homomorphisms. We have taken the case when $k$ is odd; if $k$ is even, then $\M_{2,k}^2$ and $\M_{2,k}^3$ are swapped.}
    \label{fig:minions}
\end{figure}

\noindent
Combining Theorems~\ref{thm:gadgets}, \ref{thm:sourceOfHardness}, and~\ref{thm:minions} gives the following.

\begin{corollary}\label{cor:hardness}
  $\PCSP(\LO_2^r,\LO_k^r)$ is \NP-hard for $r\geq k+2$. Moreover, there is no polynomial-time reduction using
  pp-constructions from 
  $\PCSP(\LO_2^{r'}, \LO_k^{r'})$ to
  $\PCSP(\LO_2^r, \LO_k^r)$ for
  $r'\geq k+2$ and $3\leq r<k+2$.
\end{corollary}

\noindent
The next theorem will allow us to lift NP-hardness of $\PCSP(\LO_2^r,\LO_k^r)$ to $\PCSP(\LO_{2+a}^r,\LO_{k+a}^r)$ for every positive integer $a$.

\begin{theorem}
\label{thm:colours}
For every $r\geq 3$ and $2\leq \ell<k$,  $\M_{\ell+1, k+1}^r \to \M_{\ell, k}^r$.
\end{theorem}

\begin{proof}
Consider any $p$-ary polymorphism $f \in (\M_{\ell+1, k+1}^r)^{(p)}$. Consider the value of $f$ for inputs $a_1, \ldots, a_p \in [\ell]$; due to the following matrix with $r \geq 3$ rows
\[
\begin{pmatrix}
a_1+1 & \ldots & a_p +1\\
a_1 & \ldots & a_p\\
\vdots & \ddots & \vdots \\
a_1 & \ldots & a_p\\
\end{pmatrix},
\]
we can deduce that $f(a_1, \ldots, a_p) < f(a_1 + 1, \ldots, a_p+1) \in [k + 1]$. This implies that $f(a_1, \ldots, a_p) \in [k]$. We claim this implies that $f$, restricted to $[\ell]^p$, is a polymorphism of $\M_{\ell, k}^r$. Consider matrix of inputs $a_i^j$ where $i \in [p]$, $j \in [r]$, such that each column $a_i^1, \ldots, a_i^r$ is a tuple of $\LO_\ell^r$. Thus each column is also a tuple of $\LO_{\ell+1}^r$. Since $f$ is a polymorphism of $\PCSP(\LO_{\ell+1}^r, \LO_{k+1}^r)$, we deduce that
\[
(f(a_1^1, \ldots, a_p^1), \ldots, f(a_1^r, \ldots, a_p^r))
\]
is a tuple of $\LO_{k+1}^r$ i.e.~has a unique maximum. But we already know these values belong to $[k]$. Since they have a unique maximum, they are a tuple of $\LO_k^r$. Thus $f$, restricted to $[k]^p$, is a polymorphism of $\M_{\ell, k}^r$.

We now claim that the map $f \mapsto \Crestrict{f}{[k]^p}$
taking a $p$-ary polymorphism to its restriction on $[k]^p$ is a minion homomorphism $\M_{\ell+1, k+1}^r \to \M_{\ell, k}^r$. To see why, consider any polymorphism $f \in(\M_{\ell+1, k+1}^r)^{(p)}$ and a function $\pi : [p] \to [q]$. What we need to prove is that
\[
\Crestrict{(f_\pi)}{[k]^p}= (\Crestrict{f}{{[k]^p}})_\pi.
\]
But note that, for $x_1, \ldots, x_p \in [k]$,
\begin{multline*}
\Crestrict{((f_\pi)}{{[k]^p}})(x_1, \ldots, x_p)
= f_\pi(x_1, \ldots, x_p)
= f(x_{\pi(1)}, \ldots, x_{\pi(p)}) \\
= (\Crestrict{f}{[k]^p})(x_{\pi(1)}, \ldots, x_{\pi(p)})
= (\Crestrict{f}{[k]^p})_\pi (x_1, \ldots, x_p).
\end{multline*}
This concludes the proof.
\end{proof}

Theorems~\ref{thm:gadgets}, \ref{thm:colours}, \ref{thm:sourceOfHardness}, and~\ref{thm:minions} imply the following:

\begin{corollary}\label{cor:colours}
$\PCSP(\LO_\ell^r, \LO_k^r)$ is \NP-hard for $2 \leq \ell \leq k$ and $r \geq k - \ell +4$.
\end{corollary}

\medskip
The rest of this section is devoted to the proof of Theorem~\ref{thm:minions}.

In order to construct  minion homomorphisms, as the first milestone we  exhibit a simple
necessary and sufficient condition for the existence of a minion homomorphism to
$\M_{2, k}^r$, and a sufficient condition for such a homomorphism to not exist.

\begin{lemma}\label{lem:suff}
Fix $r \geq 2$ and $ k \geq 3$. Consider any polymorphism
minion $\M$. For any element $f \in \M^{(r)}$, let $f_1(x, y) = f(y, x, \ldots, x)$, $f_2(x, y) = f(x, y, x, \ldots, x)$, \ldots, $f_r(x, y) = f(x, \ldots, x, y)$. Now, $\M \to \M_{2, k}^r$ if and only if there exists some $\omega : \M^{(2)} \to [k]$ such that, for all $f \in \M^{(r)}$, there exists a unique maximum value among $\omega(f_1), \ldots, \omega(f_r)$.
\end{lemma}
\begin{proof}
We construct $\F_\M(\LO_2^r)$. The tuples of the relation of $\LO_2^r$ are all the $r$-dimensional vectors containing exactly one $2$, with the other entries equal to $1$. We can arrange these tuples into an $r$-by-$r$ matrix where the diagonal contains 2 and all the other elements are 1. Replacing 1 with $x$ and 2 with $y$, and applying $f$, we get the definitions of $f_1, \ldots, f_r$. Thus the relation of $\F_\M(\LO_2^r)$ contains precisely the tuples of the form $(f_1, \ldots, f_r)$ for $f \in \M^{(r)}$. 

Thus our condition amounts to the existence of a homomorphism $\omega : \F_\M(\LO_2^r) \to \LO_k^r$. By Theorem~\ref{thm:adjunction}, this is equivalent to $\M \to \M_{2, k}^r$.
\end{proof}

\begin{corollary}\label{cor:suff}
Fix $r \geq 2$, $r' \geq 2$ and $k \geq 3$. If $f\{1\}, \ldots, f\{r'\}$ has a unique maximum for any function $f : 2^{[r']} \to [k]$ such that, for any partition $A_1, \ldots, A_r$ of $[r']$, $f(A_1), \ldots, f(A_r)$ has a unique maximum, then $\M_{2, k}^r \to \M_{2, k}^{r'}$.
\end{corollary}
\begin{proof}
Apply Lemma~\ref{lem:suff} using $\omega(f) = f(1, 2)$, and reinterpreting polymorphisms as indicated by Remark~\ref{rem:reinterpret}. Observe that the polymorphisms $f \in (\M_{2, k}^r)^{(r')}$ can then be seen as functions from $2^{[r']}$ to $[k]$ such that, for any partition $A_1, \ldots, A_r$ of $[r']$, a unique maximum exists among $f(A_1), \ldots, f(A_r)$. The precondition of our corollary then implies that $f\{1\}, \ldots, f\{r'\}$ have a unique maximum. By definition, $\omega(f_i) = f_i(1, 2) = f\{i\}$, where $f_i$ is defined as in Lemma~\ref{lem:suff}. Thus a unique maximum exists among $\omega(f_1), \ldots, \omega(f_r)$. It follows, by Lemma~\ref{lem:suff}, that $\M_{2,k}^r \to \M_{2,k}^{r'}$.
\end{proof}

As a second milestone, we introduce the notion of $(k,p)$-edge co-colouring and
show in Corollary~\ref{cor:sparse2} and Lemma~\ref{lem:sparse1} that cliques
have no co-colourings with certain parameters. For this, we will need three
technical lemmata, namely Lemmata~\ref{lem:3uniform},~\ref{lem:starCycle}, and~\ref{lem:graphs2}.

\begin{definition}
A $(k,p)$-edge co-colouring
of a graph $G$ is an assignment of $k$ colours to the edges of $G$ such that any $p$ disjoint edges of $G$ are not assigned the same colour. Two edges are considered disjoint if the sets of their endpoints are disjoint.

The ``co-'' prefix is included since, in such a colouring, we colour \emph{disjoint} edges with different colours, as opposed to incident edges as with normal edge colouring.
\end{definition}

\begin{lemma}\label{lem:3uniform}
A graph with $k + 3$ vertices and $3k$ edges has an independent set of size 2.
\end{lemma}
\begin{proof}
There are $(k+3)(k+2) / 2$ possible independent sets of size 2. Each edge eliminates one of them. Since $(k+3)(k+2) / 2 > 3k$ for $k \in \mathbb{R}$, it follows that at least one independent set of size 2 remains after adding in all the edges.
\end{proof}

\begin{lemma}\label{lem:starCycle}
A graph with $n$ vertices where all pairs of edges intersect is either a cycle of length 3 or a star graph (ignoring vertices with no neighbours).
\end{lemma}
\begin{proof}
This is immediately true if there are no edges, so suppose at least one edge,
  say $\{x, y\}$, exists. All other edges must intersect $\{x, y\}$, so they
  must be of the form $\{x, a\}$ or $\{y, b\}$ for some $\{a, b\}$. If all the
  edges are of the form $\{x, a\}$ or $\{y, b\}$ respectively, then the graph is a
  star, as required. Otherwise, there exist two edges $\{x, a\}$, $\{y, b\}$. For
  these edges to intersect, it must be the case that $a = b$; thus we have found
  the cycle $x, y, a = b$ in our graph. To show that only this cycle exists
  within our graph, consider any edge $\{u, v\}$ in our graph. It must intersect
  $\{x, y\}$, so we can assume, without loss of generality, that $u = x$. It must
  also intersect $\{y, b\}$, so we deduce that $v \in \{ y, b\}$. Thus our edge must already exist within the cycle.
\end{proof}

\begin{lemma}\label{lem:graphs2}
If $G$ is a graph with $k + 3$ vertices with a $(k,2)$-edge co-colouring, then $G$ has an independent set of size 2.
\end{lemma}
\begin{proof}
We prove this fact inductively. The result is immediate for $k = 0$ (in which case the graph has no edges). 

Suppose a colour $c$ exists for which a vertex exists that is adjacent to all edges of colour $c$. Removing this vertex from the graph and applying the inductive hypothesis gives us the required independent set. Thus suppose that no vertex exists that covers all edges of a particular colour. In this case, by Lemma \ref{lem:starCycle}, it follows that, for each colour, the edges with that colour form a cycle of length 3. Thus our graph has $3k$ edges at most. It follows by Lemma \ref{lem:3uniform} that there exists an independent set of size 2.
\end{proof}

\begin{corollary}\label{cor:sparse2}
Thus, for any $k \in \mathbb{N}$, $K_{k+3}$ has no $(k,2)$-edge co-colouring.
\end{corollary}
\begin{proof}
This is the contrapositive of Lemma~\ref{lem:graphs2}.
\end{proof}
\begin{remark}
This corollary is tight, in the sense that $K_{k+2}$ admits a $(k,2)$-edge co-colouring, viz.~colour edge $\{ x, y \}$ with $\max\{x, y, 3\}$, assuming that the vertex-set of $K_{k+2}$ is $[k+2]$. This colouring will reappear implicitly in Theorem~\ref{thm:impossible1} and Theorem~\ref{thm:impossible2}.
\end{remark}

\begin{lemma}\label{lem:sparse1}
For any $k \geq 3$, $K_{10k}$ has no $(k,3)$-edge co-colouring.
\end{lemma}
\begin{proof}
Suppose such a colouring exists, and let $c$ be the most frequent colour. Since there are $10k(10k-1)/2$ edges in total, and $k$ colours, it follows that there are at least $5(10k-1) \geq 45k$ edges of colour $c$. Select an edge $e$ from these, removing all edges that share an endpoint with it (of which there are at most $20k$). There now remain at least $25k$ edges. Do the same thing again, selecting an edge $e'$, and leaving at least $5k$ edges. Select $e''$ from among these edges. The edges $e, e', e''$ contradict the fact that $K_{10k}$ is a $(k,3)$-edge co-colouring.
\end{proof}

As the third milestone, using the developed results so far, we will establish,
firstly, minion homomorphisms in
Theorems~\ref{thm:minion+1},~\ref{thm:minion-2}, and~\ref{thm:minion+2}, and,
secondly, the lack of minion homomorphisms in Theorems~\ref{thm:impossible1} and~\ref{thm:impossible2}.

\begin{theorem}\label{thm:minion+1}
For $k \geq 3$, $r \geq 10k+1$, $\M_{2, k}^r \to \M_{2, k}^{r+1}$.
\end{theorem}
\begin{proof}
Consider any function $f : 2^{[r+1]} \to [k]$ such that, for any partition $A_1, \ldots, A_r$ of $[r+1]$, $f(A_1), \ldots, f(A_r)$ has a unique maximum. By Corollary \ref{cor:suff}, it is sufficient to show that $f\{1\}, \ldots, f\{r+1\}$ has a unique maximum.

Consider all pairs $\{ x, y \}$ such that the unique maximum of $f \{x, y \}, f \{1 \}, \ldots, f \{ r + 1 \}$ is given by $f \{x, y \}$. Construct a graph $G$ whose vertices are $[r+1]$, and whose edges are given by all such pairs $\{ x, y \}$. Observe that $f$ is a $(k, 3)$-edge co-colouring of $G$. This is because, for any disjoint edges $e, e', e''$, due to the partition formed out of $e, e', e''$, the empty set twice and all remaining singletons, there must exist a unique maximum cost edge among $e, e', e''$.

Now suppose that no unique maximum exists among $f\{1\}, \ldots, f\{r+1\}$. Without loss of generality, let $f\{r+1\} = f\{r\} = c$, and $f\{i\} \leq c$ for $i \in [r+1]$. Now note that, for any distinct $x, y \in [r-1]$, $f\{x, y\}$ must be the unique maximum of $f\{x, y\}, f\{1\}, \ldots, f\{r+1\}$. Thus the subgraph of $G$ given by restricting to $[r-1]$ is a complete graph --- in fact, as $r - 1 \geq 10k$, the subgraph $G$ given by restricting to $[10k]$ is a complete graph. Due to $f$ it is also $(k,3)$-edge co-colourable. But this contradicts Lemma \ref{lem:sparse1}. Thus we deduce that $f\{1\}, \ldots, f\{r+1\}$ has a unique maximum.
\end{proof}

\begin{theorem}\label{thm:minion-2}
For any $r \geq 2, k \geq 3$, $\M_{2, k}^{r+2} \to \M_{2, k}^r$.
\end{theorem}
\begin{proof}
Consider any function $f : 2^{[r]} \to [k]$ such that, for any partition $A_1, \ldots, A_{r+2}$ of $[r]$, $f(A_1), \ldots, f(A_{r+2})$ has a unique maximum. By Corollary \ref{cor:suff}, it is sufficient to show that $f\{1\}, \ldots, f\{r\}$ has a unique maximum. But consider the partition $\{1\}, \ldots, \{r\}, \emptyset, \emptyset$, and note that the largest value cannot be $f(\emptyset)$ (since $f(\emptyset)$ appears twice). Thus we deduce that one of $f\{1\}, \ldots, f\{r\}$ is the maximum, and furthermore that this maximum is strictly larger than all the other values in this sequence. Thus, $\M_{2,k}^{r+2} \to \M_{2,k}^r$.
\end{proof}

\begin{theorem}\label{thm:minion+2}
For $k \geq 3, r \geq k + 2$, $\M_{2, k}^r \to \M_{2, k}^{r+2}$.
\end{theorem}
\begin{proof}
Let $f : 2^{[r+2]} \to [k]$ be a function such that, for any partition $A_1, \ldots, A_r$ of $[r+2]$, $f(A_1), \ldots, f(A_r)$ has a unique maximum. By Corollary \ref{cor:suff}, it is sufficient to show that $f\{1\}, \ldots, f\{r+2\}$ has a unique maximum. Assume for contradiction that it does not. Then we have two cases depending on the maximum cost value among them.

\begin{description}
\item[Maximum is 1.] In this case, construct a complete graph $G$ on $[r+2]$, and assign edge $\{x, y\}$ cost $f\{x, y\}$. Observe that this $G$ is $(k,2)$-edge co-colourable, since, for any disjoint edges $\{x, y\}, \{z, t\}$, due to partition $\{x, y\}, \{z, t\}, \{1\}, \ldots, \{r+2\}$, one of the edges is assigned the unique maximum value by $f$, and thus the edges are assigned distinct costs. This contradicts Corollary~\ref{cor:sparse2}, as $r + 2 \geq k + 3$, and thus a subgraph of $G$ with $k + 3$ vertices is a $(k,2)$-edge co-colourable copy of $K_{k+3}$.

\item[Maximum is not 1.] In this case, suppose the maximum is $c > 1$, and suppose that $f\{r+1\} = f\{r+2\} = c$. Construct a complete graph $G$ on $[r]$, where edge $\{x, y\}$ is assigned cost $\max(f\{x, y\}, c)$. Observe that this $G$ is $(k-1, 2)$-edge co-colourable, since for any disjoint edges $\{x, y\}, \{z, t\}$, due to the partition $\{x, y\}, \{z, t\}, \{1\}, \ldots, \{r+2\}$, these edges are assigned distinct costs, one of which is strictly greater than $c$ (and thus these costs remain distinct after the operation $x \mapsto \max(x, c)$). This contradicts Corollary~\ref{cor:sparse2}, as $r \geq k + 2$,
and thus a subgraph of $G$ with $k + 2$ vertices is a $(k-1, 2)$-edge co-colourable copy of $K_{k+2}$.
\end{description}
Thus, since a contradiction is found in all cases, $\M_{2, k} \to \M_{2,k}^{r+2}$.
\end{proof}

\begin{theorem}\label{thm:impossible1}
For $k \geq 3$, $\M_{2, k}^{k} \not \to \M_{2, k}^{k+1}$.
\end{theorem}
\begin{proof}
Let $v(x) = \max(2, x - 1)$. Let $f : 2^{[k+1]} \to [k]$ be a function that maps the empty set and singletons to 1, sets of size $k$ or $k + 1$ to $k$, and any other set $S$ to $\max_{x \in S} v(x)$. Observe that $f_1 = \ldots = f_{k+1}$ (all of them map $(1, 1)$ and $(1, 2)$ to 1, and $(2, 1)$ and $(2, 2)$ to $k$). Thus if $f$ is a polymorphism, by Lemma~\ref{lem:suff}, and using the interpretation of Remark~\ref{rem:reinterpret}, our conclusion follows.

To see why $f$ is a polymorphism consider any partition $A_1, \ldots, A_k$ of $[k+1]$. If the partition has one part of size $k$ or $k+1$, then it has only one such part, and the other parts are either singletons or empty sets --- thus in this case $f(A_1), \ldots, f(A_k)$ contains one $k$ and $k - 1$ ones. Otherwise, observe that the partition must contain at least one part of size 2 or higher. The cost of this part will be the maximum element in the part minus 1, or 2, whichever is higher. All other parts of size 1 or 0 are assigned 1. The only way in which this partition could lack a unique maximum element is if there are two parts of size 2 or above, both of whose maximum values through $v$ are at most 2. This is impossible by the pigeonhole principle, since there are only 3 elements mapped by $v$ to 1 or 2.
\end{proof}

\begin{theorem}\label{thm:impossible2}
For $k \geq 3$, $\M_{2, k}^{k + 1} \not \to \M_{2, k}^k$.
\end{theorem}
\begin{proof}
Let $v(x) = \max(3, x)$. Let $f : 2^{[k]} \to [k]$ be a function that maps singletons to 1, the empty set to 2, sets of size $k$ or $k + 1$ to $k$, and any other set $S$ to $\max_{x \in S} v(x)$. Observe that $f_1 = \ldots = f_k$ (all of them map $(1, 1)$ to 2, $(1, 2)$ to 1, and $(2, 1)$ and $(2, 2)$ to $k$). Thus if $f$ is a polymorphism, by Lemma~\ref{lem:suff}, and using the interpretation of Remark~\ref{rem:reinterpret}, our conclusion follows.

To see why $f$ is a polymorphism consider any partition $A_1, \ldots, A_{k+1}$ of $[k]$. If the partition has one part of size $k$ or $k+1$, then it has only one such part, and the other parts are either singletons or empty sets --- thus in this case $f(A_1), \ldots, f(A_{k+1})$ contains one $k$ and $k$ ones.
Otherwise, observe that the partition must either contain at least one part of size 2 or higher, or it must contain the empty set and all singletons. In the latter case, the empty set supplies the unique maximum. In the former case, the cost of a non-singleton part will be the maximum element in the part or 3, whichever is higher. All parts of size 1 or 0 are assigned 1 or 2. The only way in which this partition could lack a unique maximum element is if there are two parts of size 2 or above, both of whose maximum values through $v$ are at most 3. This is impossible, by the pigeonhole principle, since there are at most 3 elements mapped by $v$ to 1, 2 or 3.
\end{proof}

Finally, we will now combine Theorems~\ref{thm:minion+1}--\ref{thm:impossible2} to establish Theorem~\ref{thm:minions}.

\begin{proof}[Proof of Theorem~\ref{thm:minions}]
Fix $k \geq 3$. Consider any $r \geq k+2$. Letting $r'$ be any number greater or equal to $10k$ with the same parity as $r$, note that, by applying Theorem~\ref{thm:minion+2}, Theorem~\ref{thm:minion+1} and Theorem~\ref{thm:minion-2}, we have
\[
\M_{2,k}^r \to \M_{2,k}^{r+2} \to \ldots \to
\M_{2,k}^{r'-2}\to\M_{2,k}^{r'} \to \M_{2,k}^{r'+1} \to
\M_{2,k}^{r'-1} \to \ldots \to \M_{2,k}^{r+3} \to \M_{2,k}^{r+1}.
\]
Thus, for $r \geq k+2$, $\M_{2,k}^r \to \M_{2, k}^{r+1}$. Furthermore, by applying Theorem~\ref{thm:minion-2},
\[
\M_{2, k}^r \to \M_{2,k}^{r-2} \to \ldots,
\]
and
\[
\M_{2, k}^r \to \M_{2,k}^{r+1} \to \M_{2,k}^{r-1} \to \M_{2,k}^{r-3} \ldots.
\]
Thus we deduce that, for any $r'$ such that $2 \leq r' \leq r$, $\M_{2,k}^r \to \M_{2,k}^{r'}$. This supplies all the minion homomorphisms required by Theorem~\ref{thm:minions}.

Now, consider any $2 \leq r < k + 2$. Note that by Theorem~\ref{thm:impossible1} and Theorem~\ref{thm:impossible2}, we know that $\M_{2,k}^{k+1} \not \to \M_{2,k}^k$ and $\M_{2,k}^k \not \to \M_{2,k}^{k+1}$. Note that $\M_{2,k}^{k+2} \to \M_{2, k}^k$ and $\M_{2,k}^{k+2} \to \M_{2,k}^{k+1}$. Thus by contrapositive, $\M_{2, k}^k \not \to \M_{2, k}^{k+2}$ and $\M_{2, k}^{k+1} \not \to \M_{2, k}^{k+2}$. Now, depending on the parity of $r$, by Theorem~\ref{thm:minion-2}, $\M_{2, k}^k \to \M_{2,k}^r$ or $\M_{2, k}^{k+1} \to \M_{2, k}^r$. Thus in either case, by contrapositive, $\M_{2, k}^r \not \to \M_{2, k}^{k+1}$. This supplies all the homomorphism nonexistence proofs required by Theorem~\ref{thm:minions}.
\end{proof}

\section{Conclusions}

The question about the complexity of $\PCSP(\LO^3_2,\LO^3_k)$ for constant $k\geq 3$ 
raised in~\cite{Barto21:stacs} stays open. More generally, the complexity of $\PCSP(\LO^r_\ell,\LO^r_k)$ is open except for the hardness results obtained in this paper:
We established \NP-hardness for every constant $2\leq k\leq \ell$ and every constant uniformity $r\geq \ell-k+4$. 

Minion homomorphisms (and lack thereof) between the polymorphism minions $\M_{2, k}^r$ for various values of $r$ have interesting implications for the complexity of PCSPs more broadly, beyond our hardness results. 
In particular, if one were to prove \NP-hardness of LO 2- vs. $k$-colourings on $r$-uniform hypergraphs for $r<k+2$ , then our results imply that this 
\NP-hardness does \emph{not} follow from \NP-hardness of the same problem with uniformity at least $k+2$ via minion homomorphisms and thus in particular cannot be obtained from the latter problem via ``gadget reductions''~\cite{BBKO21}. This is in contrast to the case of (non-promise) CSPs, where it is known~\cite{BOP18} (cf. also~\cite{BKW17})\footnote{\cite{BKW17} uses the terminology of height~$1$ identities.} that all 
\NP-hardness can be shown using minion homomorphisms.\footnote{This would not be the first occurrence of this phenomenon in the context of PCSPs; a recent example of the same phenomenon comes from~\cite{WZ20} for the problem of approximate graph colouring.}

\medskip
Going beyond the realm of fixed-template PCSPs~\cite{BBKO21} (which limits the number of colours by a constant), what is the smallest function $k(n)$ for which $\PCSP(\LO^3_2,\LO^3_{k(n)})$ is solvable efficiently? There is no clear reason to believe that positive result from the present paper with $k(n)=O(\sqrt[3]{n\log\log n/\log n})$ is optimal.

\paragraph{Acknowledgements}
We thank the anonymous reviewers of both the ICALP version~\cite{NZ22:icalp} and this full version for their comments. 
We also thank  D\"om\"ot\"or P\'alv\"olgyi
for informing us that LO colourings have been studied under the name of \emph{unique maximum} colourings.

{\small
\bibliographystyle{plainurl}
\bibliography{nz}
}

\end{document}